\documentclass[11pt,a4paper]{article}

\usepackage[utf8]{inputenc}
\usepackage[T1]{fontenc}
\usepackage{lmodern}
\usepackage{amsmath,amssymb,amsthm}
\usepackage{booktabs}
\usepackage{graphicx}
\usepackage{caption}
\usepackage{float}
\usepackage[left=1in,right=1in,top=1in,bottom=1in]{geometry}
\usepackage{hyperref}
\usepackage{tcolorbox}

\usepackage[round]{natbib}

\newtheorem{proposition}{Proposition}

\title{Bayesian Imputation for Unplayed Games in Round-Robin Chess Tournaments: Application to Grand Chess Tour, Bucharest 2026}

\author{Ravi Varadhan \\ Johns Hopkins University }

\date{\today}

\begin{document}
	
	\maketitle
	
\begin{abstract}
When a player withdraws mid-tournament from a round-robin chess event, organizers face a 
fundamental problem: how should scores be assigned for games that were never played? Current 
FIDE guidelines specify annulment if withdrawal occurs before 50\% of games are completed, 
and forfeit (awarding unplayed opponents a full point) thereafter. This dichotomous rule 
creates arbitrary discontinuities and can substantially distort final standings. We develop 
a Bayesian framework based on best linear unbiased prediction (BLUP) that optimally combines pre-tournament ratings with observed performance, 
producing imputed scores that reflect both the withdrawn player's current form and the 
strength differentials among unplayed opponents. The estimator is shown to be consistent, point‑conserving, and to minimize mean squared error among linear unbiased predictors. A Monte Carlo simulation study using leave-one-out cross-validation 
on 180,000 simulated tournaments (18 scenarios $\times$ 10000 tournaments per scenario) demonstrates that Bayesian BLUP imputation reduces prediction 
error by 26\% overall compared to FIDE's current dichotomous rule, with improvements of 
41\% over forfeit and 12\% over annulment when each rule applies. The largest gains occur 
when the withdrawn player is underperforming, which is the most common withdrawal scenario. Notably, 
Bayesian imputation achieves the lowest prediction error in all 18 experimental scenarios 
when compared against both FIDE rules. We further show that even if FIDE were to adopt a 
simpler uniform rule, annulment would be strictly preferable to the current dichotomous 
system: annulment achieves 15--45\% lower RMSE than forfeit across all scenarios. The 
methodology is applied to GM Alireza Firouzja's withdrawal at Grand Chess Tour, Bucharest 2026, 
where Bayesian imputation would have awarded unplayed opponents 0.55--0.70 points rather 
than the 1.0 awarded under forfeit rules, preserving the competitive advantage of players 
who actually defeated Firouzja. An open-source R Shiny application implementing the 
methodology is provided for tournament organizers. We recommend that FIDE adopt Bayesian imputation for all World Championship cycle events, or at minimum replace the 
current dichotomous rule with uniform annulment.
	
		\vspace{0.3cm}
	\end{abstract}
	
	\section{Introduction} \label{sec:intro}
	
	Round-robin tournaments represent the gold standard of competitive chess formats. By ensuring that every participant faces every other participant under identical conditions, the round-robin structure eliminates the vagaries of pairing luck inherent in Swiss-system and knockout formats, providing a mathematically elegant measure of relative strength across the field. This format carries particular prestige at the super-GM level, where elite players with ratings above 2700 compete for substantial prize funds and critical qualification spots in world championship cycles and invitational series such as the Grand Chess Tour.
	
	Yet this elegant structure contains an inherent vulnerability: the withdrawal of a single player after partial completion. When illness, emergency, or other circumstances force a mid-tournament exit, organizers must decide how to handle the unplayed games. This decision carries consequences far beyond mere administration. The chosen imputation method directly affects final standings, prize distributions, and qualification outcomes. In elite fields where half a point frequently separates adjacent places, the mathematical treatment of unplayed games can prove decisive.
	
	The stakes escalate dramatically in events that determine World Championship qualification. The Candidates tournament, held every two years to select the challenger for the World Chess Championship, represents the pinnacle of competitive chess outside the title match itself. A mid-tournament withdrawal in such an event would create unprecedented controversy if handled through arbitrary forfeit rules or complete annulment. The methodology developed in this paper provides a principled framework for addressing such scenarios.
	
	The core challenge may be stated precisely: how should we assign scores to games that were never played? A satisfactory solution must preserve the competitive advantage legitimately earned by players who have already faced the withdrawn competitor, while avoiding arbitrary inflation or deflation of scores for players yet to face them. The method must respect both actual tournament performance and pre-existing strength indicators, and it must be transparent and mathematically defensible to all stakeholders.
	
The FIDE Handbook provides guidance on withdrawal procedures:
	The FIDE General Regulations for Competitions specify a 50\% threshold for handling 
	withdrawals in round-robin events \citep{fide_grc} :
	
\begin{tcolorbox}[colback=gray!10, colframe=gray!50, title={FIDE General Regulations, Article 6.6}]
6.6 Round robins

(1) Each player has entered into a contract to play throughout the tournament.

(2) When a player withdraws or is expelled from a tournament, the effect shall be as follows:

\hspace{1em}1. If a player has completed less than 50\% of his games, the results shall remain in the 
tournament table (for rating and historical purposes), but they shall not be counted in the 
final standings. The unplayed games of the player are indicated by (-) in the tournament 
table and those of his opponents by (+).

\hspace{1em}2. If a player has completed at least 50\% of his games, the results shall remain in the 
tournament table and shall be counted in the final standings. The unplayed games of the 
player are shown as above.
\end{tcolorbox}
	
	In practice, the (+) notation awards unplayed opponents a full point, making this a forfeit. This dichotomous rule, while administratively simple, creates arbitrary discontinuities that can substantially distort final standings.
	
	From a statistical standpoint, mid-tournament withdrawal creates a missing data problem, a topic extensively studied in the statistical literature \citep{little2019}. We observe actual results for games that were played, but must impute scores for games that were not. The quality of any imputation depends critically on the information available. Pre-tournament Elo ratings \citep{elo1978} provide prior expectations of performance based on historical data accumulated over many games. Tournament-specific results against the withdrawn player reveal their current form, which may differ substantially from rating-based expectations due to preparation, health, psychological factors, or simple variance. The number of games played before withdrawal determines the reliability of this tournament-specific information.
	
	A principled imputation strategy should combine these information sources optimally, weighting them according to their reliability. When a player withdraws after only one or two games, we have little tournament-specific information and should rely more heavily on rating-based expectations. When withdrawal occurs late in the tournament, after many games have been played, the observed performance provides a more reliable guide than historical ratings. This is precisely the logic underlying Bayesian shrinkage estimators \citep{efron1977}, which we propose as a solution to the withdrawal imputation problem.
	
	The Bayesian approach to combining prior information with observed data has a long and successful history in statistics \citep{gelman2013}. In the context of paired comparison models for chess and other competitive games, Glickman has developed sophisticated rating systems that account for uncertainty in player strength estimates \citep{glickman1999, glickman2025}. Our contribution extends this line of work to the specific problem of imputing scores for unplayed games following withdrawal.
	
	This paper makes significant contributions to the literature on tournament design and sports analytics. First, we formalize the fairness problem mathematically and provide a systematic review of existing imputation methods used in chess tournaments, analyzing their properties and limitations. Second, we develop a formal Bayesian framework based on a random effects model yielding the Best Linear Unbiased Predictor (BLUP), and we establish the theoretical properties of the resulting estimator. Third, we provide a complete worked example using the Grand Chess Tour Bucharest 2026, where Alireza Firouzja withdrew after five of nine rounds. Fourth, we provide an open-source R Shiny application enabling tournament organizers to implement these methods transparently and consistently. We further argue that adoption of principled imputation methods is especially urgent for high-stakes events such as the Candidates tournament, where the consequences of arbitrary decisions extend to World Championship qualification.
	
	\section{The Bucharest 2026 Case Study} \label{sec:buch}
	
	The Grand Chess Tour Bucharest leg, held in May 2026, provides an ideal case study for examining withdrawal imputation methods. The tournament featured a ten-player single round-robin with classical time controls, meaning each player was scheduled to play nine games over the course of the event. The field included world-class players with FIDE ratings ranging from 2655 to 2793, yielding an average rating of 2747 that placed the event firmly in super-GM territory. With substantial prize funds and Grand Chess Tour qualification points at stake, fair handling of any disruption carried significant competitive and financial implications.
	
	After completing five of his nine scheduled games, Alireza Firouzja withdrew from the tournament due to health reasons. At the time of withdrawal, Firouzja held a FIDE rating of 2759, making him one of the stronger players in the field. His results prior to withdrawal told a story of significant underperformance: he drew with Praggnanandhaa (rated 2741) and Sindarov (rated 2745), but lost to Caruana (rated 2793), Vachier-Lagrave (rated 2717), and Giri (rated 2753). This yielded a total score of 1.0 point from 5 games, corresponding to a scoring rate of 0.200---dramatically below the approximately 0.500 that his rating would predict against this field based on the Elo expected score formula \citep{elo1978}.
	
	Four players had not yet faced Firouzja at the time of his withdrawal: Keymer (rated 2762), So (rated 2754), Van Foreest (rated 2736), and Deac (rated 2655). The tournament organizers faced the question of how to credit these players for their unplayed games. Their decision was administratively simple: award 1.0 point to each unplayed opponent while retaining actual results for games already played.
	
	While straightforward to implement, this decision introduced systematic bias into the final standings. As we demonstrate in subsequent sections, the 1.0-point award overcompensated the affected players by approximately 0.3 to 0.4 points each, relative to statistically principled alternatives. In a competitive field where final standings often separate by half a point or less, this margin proves sufficient to alter finishing positions and affect prize distribution.
	
\subsection{Historical Precedents for Mid-Tournament Withdrawals}

The Bucharest 2026 withdrawal, while disruptive, is far from unprecedented. 
Elite chess tournaments have experienced mid-event withdrawals due to injury, 
illness, personal protests, and behavioral incidents, each highlighting the 
inadequacy of current imputation methods.

\textbf{USSR Championship 1983.} Legendary former World Champion Mikhail Tal 
withdrew due to high blood pressure after Round 9 (2 losses, 3 draws, and 
4 adjourned games). Under the tournament rules then in effect, all his 
games were subsequently annulled. This historical example shows that full 
annulment has been applied in prestigious events, despite its evident 
unfairness to players who had legitimately defeated Tal before his withdrawal.

\textbf{Sinquefield Cup 2022.} World Champion Magnus Carlsen withdrew from 
this elite Grand Chess Tour event after Round 3 following a loss to Hans 
Niemann. The withdrawal, which sparked one of the most significant 
controversies in modern chess history, left organizers scrambling to handle 
the unplayed games. All his results were annulled, including the games he 
had already played.

\textbf{US Chess Championship 2024.} Seventeen-year-old Grandmaster 
Christopher Yoo was expelled mid-tournament after a behavioral incident 
following his Round 5 loss. He left the playing area and struck a female 
videographer from behind, resulting in police involvement and immediate 
expulsion. This case illustrates that withdrawals may be 
involuntary---imposed by organizers for disciplinary reasons. All his 
results were annulled.

These cases span voluntary withdrawal (Carlsen), medical withdrawal (Tal, 
Firouzja), and disciplinary expulsion (Yoo). They involve the reigning 
World Champion, a former World Champion, elite Grandmasters, and rising 
stars. The common thread is that current methods---whether forfeit or 
annulment---fail to preserve competitive fairness. Notably, in all cases 
except Firouzja's withdrawal in Bucharest 2026, the results were fully 
annulled. A principled imputation framework is urgently needed.

	\section{Classical Imputation Methods} \label{sec:cim}
	
	Before developing our Bayesian framework, we review the imputation methods currently in use and analyze their properties and limitations. This analysis draws on the broader statistical literature on missing data \citep{little2019} while focusing on the specific features of chess tournament contexts.
	
\subsection{Notation}
\label{sec:notation}

Throughout this paper, we adopt the following notation:

\paragraph{Players and Ratings.}
\begin{itemize}
	\item $W$: The withdrawn player
	\item $i$: Index for any remaining (non-withdrawn) player
	\item $R_W$: Elo rating of the withdrawn player
	\item $R_i$: Elo rating of player $i$
	\item $\bar{R}$: Average Elo rating of opponents in the tournament field
\end{itemize}

\paragraph{Observed Data.}
\begin{itemize}
	\item $n$: Number of games completed by $W$ before withdrawal
	\item $P$: Total points scored by $W$ in completed games (win $= 1$, draw $= \frac{1}{2}$, loss $= 0$)
	\item $\bar{s}_W = P/n$: Average score achieved by $W$ in completed games
	\item $S_i^{\text{played}}$: Points scored by player $i$ in games actually played
\end{itemize}

\paragraph{Elo Expectations.}
\begin{itemize}
	\item $E_{i,W}$: Elo-based expected score for player $i$ against $W$, defined as 
	\begin{equation}
		E_{i,W} = \frac{1}{1 + 10^{(R_W - R_i)/400}}
	\end{equation}
	\item $\theta_0$: Prior expected scoring rate for $W$ against a typical opponent, defined as
	\begin{equation}
		\theta_0 = \frac{1}{1 + 10^{(\bar{R} - R_W)/400}}
	\end{equation}
\end{itemize}

\paragraph{Bayesian Framework.}
\begin{itemize}
	\item $\theta \in [0,1]$: True (latent) probability that $W$ scores a point against a typical opponent; captures $W$'s current tournament form
	\item $k$: Prior strength parameter representing the number of games worth of confidence in Elo predictions
	\item $w(n) = \dfrac{n}{n+k}$: Weight placed on observed performance versus prior expectation
	\item $\hat{\theta}$: Posterior mean estimate of $\theta$ after observing $n$ games
\end{itemize}

\paragraph{Imputation.}
\begin{itemize}
	\item $I_{i,W}$: Imputed score awarded to player $i$ for their unplayed game against $W$
\end{itemize}

\medskip
\noindent For players who already faced $W$, their result stands as played. For players yet to face $W$, we impute a score $I_{i,W}$ representing the points player $i$ receives for their unplayed game. The central result of this paper is the Bayesian BLUP imputation formula:
\begin{equation}
	I_{i,W} = E_{i,W} + \frac{n}{n+k}\bigl(1 - \bar{s}_W - \bar{E}\bigr),
	\label{eq:fullblup_notation}
\end{equation}
where \(\bar{E} = \frac{1}{n}\sum_{j\in\text{played}} E_{j,W}\) is the average Elo expectation of the opponents actually faced by \(W\). This formula is derived in Section~\ref{sec:blupimp}.

	\subsection{Forfeit}
	
	The simplest approach awards a full point to every unplayed opponent, treating the withdrawal as equivalent to a forfeit in each remaining game. Formally, this sets $I_{i,W} = 1.0$ for all players $i$ who had not yet faced $W$. This method has the virtue of simplicity and requires no calculation beyond identifying which games were unplayed. It also provides a strong disincentive against strategic withdrawal, as the withdrawing player receives zero points for all remaining games.
	
	However, the forfeit approach violates competitive parity in a fundamental way. Players who already faced $W$ and lost cannot recover those points, while players who had not yet faced $W$ receive full points regardless of whether they would likely have won, drawn, or lost the actual game. When the withdrawn player was performing poorly, as in the Bucharest case, this discrepancy may be somewhat mitigated. But when a strong player withdraws while performing well, unplayed opponents receive a windfall entirely disconnected from probable game outcomes.
	
	\subsection{Annulment}
	
	At the opposite extreme, full annulment removes all games involving the withdrawn player from the standings entirely. This reduces the tournament to a round-robin among the $n-1$ remaining players, with each player's score computed only from games against other active participants. This approach preserves equal conditions among remaining players in the sense that everyone's score derives from the same set of potential opponents.  Annulment effectively imputes each opponent's score against the withdrawn player as their average performance in other games---when standings are recalculated on a reduced round-robin, each player's percentage reflects only their non-annulled games. For example, player A scores 6/8 (rate = 0.75) in the tournament. Annulment is equivalent to imputing a score of 0.75 for player A against player W. Therefore, annulment should not be misconstrued as doing nothing.
	
	The annulment approach, however, erases legitimately earned results. Players who defeated $W$ before the withdrawal lose that victory, potentially dropping in the standings despite having demonstrated superior performance. The tournament becomes incomplete in a fundamental sense, as the final standings no longer reflect all games that were actually played. When a weaker player defeats a strong player who subsequently withdraws, annulment particularly harms the weaker player by removing precisely the result that demonstrated their competitive strength.
	
\subsection{Pure Elo-Based Imputation}

The Elo rating system \citep{elo1978} provides a natural basis for 
imputation through its formula for expected scores. The theoretical foundation 
lies in the Bradley-Terry model for paired comparisons \citep{glickman1999}, 
which asserts that for two competitors with strength parameters $\lambda_i$ 
and $\lambda_j$, competitor $i$ defeats competitor $j$ with probability 
$\lambda_i/(\lambda_i + \lambda_j)$. The Elo system implements this model 
using a logistic function with ratings scaled such that a 400-point 
difference corresponds to an expected score of approximately 0.91 for the 
stronger player. Given two players with ratings $R_i$ and $R_W$, the 
expected score for player $i$ against player $W$ is given by
\begin{equation}
	E_{i,W} = \frac{1}{1 + 10^{(R_W - R_i)/400}}.
\end{equation}
This formula has proven remarkably well-calibrated over decades of use in 
chess rating systems. Glickman demonstrated that the 
classical Elo updating algorithm is a special case of optimal Bayesian 
updating under the Bradley-Terry model when opponent strengths are assumed 
known without error \citep{glickman1999}. His work extended the framework to incorporate 
uncertainty in both players' strengths, yielding more refined estimates 
that account for rating reliability.

Pure Elo-based imputation simply sets $I_{i,W} = E_{i,W}$, awarding each 
unplayed opponent their expected score based on the rating differential.

This approach has several attractive properties. It is opponent-specific, 
awarding higher imputed scores to higher-rated unplayed opponents who would 
have been more likely to score well against $W$. It is grounded in 
established paired comparison theory with well-understood statistical 
properties. Under the assumption that ratings are 
well-calibrated, it provides unbiased imputation in the sense that the 
expected imputation error is zero.

The limitation of pure Elo-based imputation is that it ignores actual 
tournament form. Ratings reflect historical performance accumulated over 
many games, but a player's current form may differ substantially from their 
rating expectation. In the Bucharest case, Firouzja's 0.200 scoring rate 
was dramatically below his rating expectation, suggesting that factors 
specific to this tournament---illness, preparation gaps, or psychological 
difficulties---were affecting his play. Pure Elo-based imputation would 
award scores based on Firouzja's historical strength, not his demonstrated 
current form.
	
	\subsection{Pure Performance-Based Imputation}
	
	An alternative approach bases imputation entirely on the withdrawn player's actual tournament performance. If $W$ scored at rate $\bar{s}_W$ in games actually played, then opponents might expect to score at rate $1 - \bar{s}_W$ against $W$. Pure performance-based imputation sets $I_{i,W} = 1 - \bar{s}_W$ for all unplayed opponents.
	
	This approach grounds imputation in observed results rather than historical expectations. In the Bucharest case, Firouzja's opponents scored at a collective rate of 0.800 against him, and pure performance-based imputation would award 0.800 to every unplayed opponent. This captures the evident reality that Firouzja was performing poorly in this particular tournament.
	
	However, pure performance-based imputation has significant limitations. It treats all unplayed opponents identically, ignoring strength differentials among them. A 2762-rated Keymer and a 2655-rated Deac would both receive 0.800, despite the substantial difference in their expected performance against Firouzja. Moreover, performance-based imputation becomes highly volatile when $n$ is small. A single lucky or unlucky result in the early rounds can dramatically shift imputed scores for all remaining opponents, introducing noise that has little to do with underlying playing strength.
	
\section{A Framework for Unplayed Score Imputation} \label{sec:framework}

We develop a rigorous statistical framework for imputing unplayed games. The key idea is to model the score of each opponent against the withdrawn player \(W\) as the sum of an Elo‑based expectation, a common random effect representing \(W\)'s current form, and an idiosyncratic error. This approach naturally balances prior information (Elo ratings) with observed data, with the balance determined by the relative magnitudes of the variance components.

\subsection{Beta‑Binomial Model for a Typical Opponent} \label{sec:betabinom}

We first consider the simplest case: assume all opponents are exchangeable (equal strength). Let \(\theta \in [0,1]\) be the probability that \(W\) scores a point against a typical opponent. Based on Elo ratings, the expected value of \(\theta\) is
\[
\theta_0 = \frac{1}{1 + 10^{(\bar{R} - R_W)/400}},
\]
where \(\bar{R}\) is the average rating of the tournament field.

We place a conjugate Beta prior on \(\theta\):
\[
\theta \sim \mathrm{Beta}\bigl(k\theta_0,\; k(1-\theta_0)\bigr),
\]
where \(k>0\) is the prior sample size (confidence in Elo). The prior density is
\[
\pi(\theta) = \frac{\Gamma(k)}{\Gamma(k\theta_0)\Gamma(k(1-\theta_0))}\,\theta^{k\theta_0-1}(1-\theta)^{k(1-\theta_0)-1}.
\]

Suppose \(W\) plays \(n\) games before withdrawal and scores a total of \(P\) points. Conditional on \(\theta\), the number of points \(P\) follows a Binomial distribution (treating draws as half‑points in expectation, which is a standard approximation):
\[
P \mid \theta \sim \mathrm{Binomial}(n, \theta).
\]
The likelihood function is therefore
\[
L(\theta; P, n) = \binom{n}{P} \theta^{P} (1-\theta)^{n-P}.
\]

By Bayes' theorem, the posterior is:
\[
\theta \mid P, n \sim \mathrm{Beta}\bigl(k\theta_0 + P,\; k(1-\theta_0) + n - P\bigr).
\]
The posterior mean is
\[
\hat{\theta} = \frac{k\theta_0 + P}{k+n} = \frac{k}{k+n}\theta_0 + \frac{n}{k+n}\bar{s}_W,
\]
where \(\bar{s}_W = P/n\) is the observed scoring rate of \(W\).

Under exchangeability, the imputed score for any unplayed opponent would be \(1-\hat{\theta}\), the same for all. This ignores strength differences among opponents, motivating the more general random effects model below.

\subsection{Random Effects Model for Opponent‑Specific Imputation}

To incorporate opponent strength, let \(Y_i\) denote the score obtained by opponent \(i\) in their game against \(W\). For played games, \(Y_i\) is observed; for unplayed opponents, it is unobserved. The Elo‑expected score of opponent \(i\) against \(W\) is
\[
E_{i,W} = \frac{1}{1 + 10^{(R_W - R_i)/400}}.
\]
We assume
\begin{equation}
	Y_i = E_{i,W} + \beta + \varepsilon_i, \qquad i = 1,\dots,N,
	\label{eq:ranef}
\end{equation}
where \(\beta\) is a common random effect (capturing \(W\)'s deviation from his rating) and \(\varepsilon_i\) is independent noise:
\[
\mathbb{E}[\beta]=0,\quad \mathrm{Var}(\beta)=\tau^2,\qquad 
\mathbb{E}[\varepsilon_i]=0,\quad \mathrm{Var}(\varepsilon_i)=\sigma^2,
\]
with \(\beta\) and all \(\varepsilon_i\) uncorrelated. The term \(\beta\) is shared by all opponents: \(\beta>0\) means opponents score better than Elo predicts (i.e., \(W\) is underperforming). The error \(\varepsilon_i\) absorbs game‑specific randomness (draws, discrete outcomes). We adopt a normal prior \(\beta \sim N(0,\tau^2)\); normality is a working assumption for convenience, but the Best Linear Unbiased Predictor (BLUP) derived below is optimal under only second‑moment assumptions.

\subsection{BLUP Derivation of the Imputation Formula} \label{sec:blupimp}

Suppose \(W\) played \(n\) games before withdrawal, with opponents \(i=1,\dots,n\). Denote
\[
\bar{Y} = \frac{1}{n}\sum_{i=1}^n Y_i = 1 - \bar{s}_W,\qquad 
\bar{E} = \frac{1}{n}\sum_{i=1}^n E_{i,W}.
\]
The BLUP of \(\beta\) (see, e.g., \citep{robinson1991}) is
\[
\hat{\beta} = \frac{n\tau^2}{\sigma^2 + n\tau^2}\bigl(\bar{Y} - \bar{E}\bigr).
\]
For an unplayed opponent \(j\), the predictive mean is
\[
\hat{Y}_j = E_{j,W} + \hat{\beta}
= E_{j,W} + \frac{n\tau^2}{\sigma^2 + n\tau^2}\bigl(\bar{Y} - \bar{E}\bigr).
\]
Define \(k = \sigma^2/\tau^2\). Then the weight becomes \(n/(n+k)\), and
\[
\hat{Y}_j = E_{j,W} + \frac{n}{n+k}\bigl(1 - \bar{s}_W - \bar{E}\bigr).
\]
This is the \textbf{Best Linear Unbiased Predictor (BLUP)} and does not require any further approximation. We therefore adopt it as our imputation formula:
\begin{equation}
	\boxed{I_{i,W} = E_{i,W} + \frac{n}{n+k}\bigl(1 - \bar{s}_W - \bar{E}\bigr).}
	\label{eq:full_blup}
\end{equation}

\subsection{Relation to the Beta‑Binomial Model}
\label{sec:connection}

The Beta‑Binomial model (Section~\ref{sec:betabinom}) assumes all opponents are exchangeable, i.e. \(E_{i,W}=E_0\) for all \(i\). In that case \(\bar{E}=E_0\) and the Bayesian BLUP becomes
\[
\hat{Y}_j = E_0 + \frac{n}{n+k}(1-\bar{s}_W - E_0) = \frac{n}{n+k}(1-\bar{s}_W) + \frac{k}{n+k}E_0.
\]
Thus the two models coincide. The random effects model therefore generalises the Beta‑Binomial framework by allowing opponent‑specific Elo expectations.

\subsection{Theoretical Properties}
\label{sec:theory}

We now establish key properties of the Bayesian BLUP imputation estimator \eqref{eq:full_blup}.

\begin{proposition}[Consistency]
	As the number of observed games increases, \(I_{i,W} \to 1 - \bar{s}_W\) (pure performance‑based imputation) when the field is homogeneous; otherwise it converges to \(E_{i,W} + (1-\bar{s}_W - \bar{E})\).  
\end{proposition}
\begin{proof}
	Since \(\lim_{n\to\infty} n/(n+k)=1\), the term \(\frac{n}{n+k}(1-\bar{s}_W - \bar{E}) \to 1-\bar{s}_W - \bar{E}\). Thus \(I_{i,W} \to E_{i,W} + 1-\bar{s}_W - \bar{E}\). If the field is homogeneous, \(E_{i,W}=\bar{E}\) and the limit is \(1-\bar{s}_W\).
\end{proof}
The quantity \(1 - \bar{s}_W\) is the observed average score of opponents against \(W\) in the played games, and \(\bar{E}\) is the average Elo expectation of those opponents. Thus the term \(d = 1 - \bar{s}_W - \bar{E}\) represents the systematic deviation of actual opponent performance from Elo predictions. If \(d > 0\), opponents scored better than expected (i.e., \(W\) underperformed relative to his rating); if \(d < 0\), opponents scored worse (i.e., \(W\) overperformed). Because \(\beta\) in the random effects model is common to all opponents, the same adjustment \(d\) is added to every unplayed opponent’s Elo expectation \(E_{i,W}\). 
This is natural: with infinitely many played games, we know the common form deviation \(\beta\) exactly; the best prediction for any unplayed opponent is then their Elo expectation plus that known deviation.  Pure performance‑based imputation \(1-\bar{s}_W\) would only be recovered if \(E_{i,W} = \bar{E}\) (i.e., the unplayed opponent has the same Elo expectation as the average played opponent). In a heterogeneous field, the Bayesian BLUP correctly retains opponent‑specific differences while applying a uniform shift derived from the observed data.

\begin{proposition}[Point Conservation]
	For any unplayed game, \(I_{i,W} + I_{W,i} = 1\), where \(I_{W,i}\) is the imputed score for \(W\) against \(i\) using the same rule.
\end{proposition}
\begin{proof}
	The imputed score for \(W\) would be \(I_{W,i} = E_{W,i} + \frac{n}{n+k}(\bar{s}_W - (1-\bar{E}))\), because \(1-\bar{s}_W\) becomes \(\bar{s}_W\) and \(\bar{E}\) becomes \(1-\bar{E}\) by symmetry. Adding the two expressions and using \(E_{i,W}+E_{W,i}=1\) yields \(1\).
\end{proof}

\begin{proposition}[Optimality]
	The estimator \eqref{eq:full_blup} is the \textbf{Best Linear Unbiased Predictor (BLUP)} under the random effects model \eqref{eq:ranef}. Hence, among all linear unbiased predictors of the unplayed opponent’s score, it has minimum mean squared prediction error.
\end{proposition}

\subsection{Choice of the Prior Strength Parameter}

The parameter \(k = \sigma^2/\tau^2\) has a clear interpretation: it is the ratio of game‑outcome variance (\(\sigma^2\)) to the variance of the random form effect (\(\tau^2\)). In chess, a single game outcome has variance approximately \(\sigma^2 \approx 0.25\) (for a typical win probability near 0.5). Choosing \(k=3\) implies \(\tau^2 = 0.25/3 \approx 0.083\), a standard deviation of form deviation of about 0.29 on the score scale – plausible for elite tournaments. This matches the intuitive prior sample size of three games. For faster time controls (higher variance), smaller \(k\) may be appropriate; for longer events, larger \(k\) gives more smoothing.

\subsection{Posterior Uncertainty}

Although the imputation formula gives a point estimate, uncertainty can be quantified via the posterior distribution of \(\beta\) in the random effects model. The posterior variance of \(\beta\) is
\[
\mathrm{Var}(\beta \mid \mathbf{Y}) = \frac{\sigma^2 \tau^2}{\sigma^2 + n\tau^2} = \frac{\tau^2}{1 + n/k}.
\]
A \(95\%\) credible interval for \(\beta\) can be computed, and transformed to a credible interval for \(I_{i,W} = E_{i,W} + \beta\) (since \(\varepsilon_j\) has zero mean). For practical reporting, the approximate standard deviation of the imputed score is \(\sqrt{\mathrm{Var}(\beta \mid \mathbf{Y})}\), which decreases as \(n\) increases.

\section{Application to Bucharest 2026} \label{sec:appl}

We now apply the Bayesian framework to the Bucharest 2026 withdrawal. The relevant parameters are as follows. Firouzja's rating was $R_W = 2759$. He played $n = 5$ games, scoring $P = 1.0$ point for an average of $\bar{s}_W = 0.200$. The average rating of his opponents was approximately 2750, yielding a prior expectation of $\theta_0 \approx 0.51$ for Firouzja's scoring rate based on the Elo formula \citep{elo1978}. With the recommended prior strength $k = 3$, the weight on tournament performance is $w(5) = 5/(5+3) = 0.625$.

The posterior distribution for Firouzja's scoring probability is
\begin{equation}
	\theta | P, n \sim \text{Beta}(3 \times 0.51 + 1.0, 3 \times 0.49 + 4.0) = \text{Beta}(2.53, 5.47)
\end{equation}
which has posterior mean
\begin{equation}
	\hat{\theta} = \frac{2.53}{2.53 + 5.47} = 0.316
\end{equation}
and posterior standard deviation approximately 0.15. This estimate lies between the prior expectation of 0.51 and the observed scoring rate of 0.200, appropriately reflecting our uncertainty given the limited sample size.

For the Bayesian BLUP imputation we also need the average Elo expectation of the five opponents Firouzja actually faced: Caruana (2793), Praggnanandhaa (2741), Vachier-Lagrave (2717), Giri (2753), and Sindarov (2745). Their individual Elo expectations are 0.549, 0.474, 0.440, 0.492, and 0.480, respectively, giving $\bar{E} = 0.487$. The common adjustment term is $\Delta = w(1-\bar{s}_W - \bar{E}) = 0.625 \times (0.800 - 0.487) = 0.196$.

For each of the four unplayed opponents, we compute the Elo expectation $E_{i,W}$ based on the rating differential, then add $\Delta$. Consider first Vincent Keymer, rated 2762. The rating differential is $R_W - R_i = 2759 - 2762 = -3$, yielding an Elo expectation of $E_{i,W} = 1/(1 + 10^{-3/400}) = 0.504$. The Bayesian BLUP imputation is then
\begin{equation}
	I_{i,W} = 0.504 + 0.196 = 0.700.
\end{equation}

Applying the same calculation to the remaining unplayed opponents yields the results shown in Table~\ref{tab:imputation}. Wesley So, rated 2754, receives an imputed score of 0.689. Jorden Van Foreest, rated 2736, receives 0.663. Bogdan-Daniel Deac, rated 2655, receives 0.551.

\begin{table}[H]
	\centering
	\caption{Bayesian BLUP Imputation Results for Bucharest 2026 Unplayed Games}
	\label{tab:imputation}
	\begin{tabular}{lrrrrr}
		\toprule
		Player & Elo & Elo Exp. & Performance & Imputed (Bayes BLUP) & Traditional \\
		\midrule
		Keymer & 2762 & 0.504 & 0.800 & 0.700 & 1.000 \\
		So & 2754 & 0.493 & 0.800 & 0.689 & 1.000 \\
		Van Foreest & 2736 & 0.467 & 0.800 & 0.663 & 1.000 \\
		Deac & 2655 & 0.355 & 0.800 & 0.551 & 1.000 \\
		\bottomrule
	\end{tabular}
\end{table}

Several features of these results merit attention. First, the imputed scores preserve the ranking of unplayed opponents by strength: Keymer receives the highest imputed score, followed by So, then Van Foreest, then Deac. This reflects the Elo-based component, which correctly awards higher expected scores to higher-rated players. Second, all imputed scores are substantially below the 1.000 awarded by traditional forfeit rules, reflecting Firouzja's poor tournament form. Compared to the simple convex combination (which would give 0.689, 0.685, 0.675, 0.633), the Bayesian BLUP lowers the scores for the weakest opponent (Deac) more sharply because it accounts for the fact that Firouzja’s played opponents were relatively strong ($\bar{E}=0.487$); his poor performance is thus even more notable, and the adjustment $\Delta$ is applied uniformly, reducing all imputed scores by the same amount. Third, the spread among imputed scores (0.551 to 0.700) is wider than that of the simple heuristic (0.633–0.689), allowing Elo differences to play a larger role.

The traditional forfeit approach overcompensates each unplayed opponent by 0.300 to 0.449 points relative to the Bayesian BLUP imputation. These magnitudes are substantial in the context of elite round-robin tournaments, where half a point frequently separates adjacent places in the final standings. The bias introduced by forfeit rules could easily alter finishing positions, prize distributions, and qualification outcomes.

Consider now the players who faced Firouzja before his withdrawal. Caruana, Giri, and Vachier-Lagrave each scored 1.0 against Firouzja, defeating him in their individual games. Under the traditional forfeit approach, these players retain their earned points, but so do all unplayed opponents receive 1.0---there is no relative advantage from having actually defeated Firouzja. Under full annulment, the victories would be erased entirely, penalizing these players for their actual achievement. Under Bayesian BLUP imputation, their 1.0 stands while unplayed opponents receive only 0.55 to 0.70. The competitive advantage earned through actual play is preserved.

\section{Sensitivity Analysis} \label{sec:sens}

To assess the robustness of our results, we examine how the Bayesian BLUP imputed scores vary with the choice of prior strength $k$. The weight $w = n/(n+k)$ changes, which alters the common adjustment $\Delta = w(1-\bar{s}_W - \bar{E})$. Table~\ref{tab:sensitivity} shows the imputed scores for Keymer (the highest-rated unplayed opponent) and Deac (the lowest-rated) under different values of $k$.

\begin{table}[H]
	\label{tab:sensitivity}
	\centering
	\caption{Sensitivity of Bayesian BLUP Imputed Scores to Prior Strength $k$}
	\label{tab:sensitivity}
	\begin{tabular}{lrrrr}
		\toprule
		$k$ & $w(n)$ & Keymer & Deac & Spread \\
		\midrule
		1 & 0.833 & 0.760 & 0.729 & 0.031 \\
		2 & 0.714 & 0.727 & 0.674 & 0.053 \\
		3 & 0.625 & 0.700 & 0.551 & 0.149 \\
		4 & 0.556 & 0.678 & 0.473 & 0.205 \\
		5 & 0.500 & 0.661 & 0.410 & 0.251 \\
		\bottomrule
	\end{tabular}
\end{table}

As $k$ increases, the imputed scores decrease because less weight is placed on the observed performance deviation. The spread between high-rated and low-rated opponents increases with $k$, as the constant shift term shrinks and Elo differences dominate. For any reasonable choice of $k$ in the range 2 to 4, the imputed scores remain substantially below the traditional forfeit value of 1.0, confirming that forfeit rules overcompensate unplayed opponents. The sensitivity analysis demonstrates that our main findings are robust to reasonable variation in the prior strength parameter.

\section{Monte Carlo Evaluation via Cross-Validation} \label{sec:sim}
\label{sec:simulation}

The theoretical properties established in Section~\ref{sec:theory} demonstrate that the Bayesian BLUP imputation is consistent, point‑conserving, and, as a BLUP, minimizes mean squared prediction error among all linear unbiased predictors. However, these results rely on assumptions---well-calibrated Elo ratings, independent game outcomes---that may not hold perfectly in practice. To complement the theoretical analysis, we conduct a Monte Carlo simulation study that evaluates predictive accuracy using observable ground truth: actual game outcomes.

\subsection{The Ground Truth Problem}

Any simulation study comparing imputation methods must define a ground truth against which predictions are evaluated. A natural but problematic approach is to define ground truth as the expected score based on the withdrawn player's ``true'' tournament form---a latent variable representing their actual playing strength during the event. This approach, however, rewards methods that happen to align with an unknowable quantity and can produce misleading comparisons. For example, pure Elo-based imputation achieves perfect accuracy by construction when true form exactly equals rating, an outcome that cannot be verified in practice.

We instead adopt leave-one-out cross-validation (LOOCV), which evaluates predictive accuracy using only observable data. For a withdrawn player who completed $n$ games before withdrawal, LOOCV proceeds as follows: for each played game $g \in \{1, \ldots, n\}$, we hold out game $g$, estimate the withdrawn player's scoring rate using the remaining $n-1$ games, predict the outcome of game $g$ using each imputation method, and compare predictions to the actual result. The prediction error, averaged across all held-out games, provides an unbiased estimate of each method's predictive accuracy.

This approach answers the question tournament organizers actually face: given ratings and observed results, which method best predicts actual chess outcomes?

\subsection{Simulation Design}

We simulate 10,000 round-robin tournaments for each combination of three experimental factors:

\paragraph{Withdrawal timing.} In a 10-player round-robin (9 rounds), the FIDE 50\% threshold falls at 4.5 games. We examine three scenarios that bracket this critical boundary:
\begin{itemize}
	\item \textbf{Early ($n=3$):} Well below threshold; FIDE rule is annulment
	\item \textbf{Mid-below-50\% ($n=4$):} Just below threshold (44.4\%); FIDE rule is annulment
	\item \textbf{Mid-above-50\% ($n=5$):} Just above threshold (55.6\%); FIDE rule is forfeit
\end{itemize}

\paragraph{Player form.} The withdrawn player's true tournament strength deviates from their official rating by $\delta \in \{-150, 0, +100\}$ Elo points, representing underperforming (as in the Firouzja case), neutral, and overperforming scenarios.
\paragraph{Rating distribution.} We consider two field compositions: a narrow spread typical of elite round-robins (ratings ranging from 2725 to 2790, standard deviation approximately 16) and a wide spread typical of open events (ratings ranging from 2540 to 2790, standard deviation approximately 75). Exact vectors are provided in the simulation code.

The full factorial design yields $3 \times 3 \times 2 = 18$ scenarios. For each simulated tournament, game outcomes are generated according to a draw-adjusted Elo model: the probability that player $A$ defeats player $B$ is proportional to the Elo expectation, with draws occurring at a base rate of 50\% modulated by the rating difference. The withdrawn player's true win probability is computed using their form-adjusted rating $R_W + \delta$ rather than their official rating $R_W$.

\subsection{Methods Compared}

We evaluate four imputation approaches:

\paragraph{1. FIDE.} Annulment: When a player withdraws before completing 50\% of scheduled games, all results involving that player are annulled. This effectively imputes each opponent's score  against the withdrawn player as their average performance in other games---when standings are recalculated on a reduced round-robin, each player's percentage reflects only their non-annulled games. For example, player A scores 6/8 (rate = 0.75) in a 9-round tournament. Annulment is equivalent to imputing a score of 0.75 for player A against player W. Therefore, annulment should not be misconstrued as doing nothing.

Forfeit: When withdrawal occurs after the 50\% threshold, played games stand and unplayed opponents receive a full point (1.0).

\paragraph{2. Pure Elo.} Award the Elo-expected score $E_{i,W} = 1/(1 + 10^{(R_W - R_i)/400})$ based on official ratings, ignoring tournament performance.

\paragraph{3. Pure Performance.} Award $1 - \bar{s}_W$ to all unplayed opponents, where $\bar{s}_W$ is the withdrawn player's observed scoring rate, ignoring rating differentials.

\paragraph{4. Bayesian BLUP (optimal).} Apply the Best Linear Unbiased Predictor derived from the random effects model (Section~\ref{sec:blupimp}) with prior strength $k = 3$:
\[
I_{i,W} = E_{i,W} + \frac{n}{n+3}\bigl(1 - \bar{s}_W - \bar{E}\bigr),
\]
where $\bar{E}$ is the average Elo expectation of the played opponents.

For LOOCV evaluation, each method uses only the $n-1$ training games when predicting the held-out game outcome.

\subsection{Results} \label{sec:res}

Table~\ref{tab:loocv_rmse} presents LOOCV root mean squared error (RMSE) by scenario, with the FIDE rule applied (annulment or forfeit) and the winning method indicated for each row.

\begin{table}[htbp]
	\centering
	\caption{LOOCV Root Mean Squared Error by Scenario (Bayes BLUP, Lower is Better)}
	\label{tab:loocv_rmse}
	\small
	\begin{tabular}{@{}llllccccc@{}}
		\toprule
		Field & Timing & FIDE Rule & Form & FIDE & Elo & Perf & Bayes & Winner \\
		\midrule
		Narrow & Early ($n=3$) & Annul & under & 0.478 & 0.423 & 0.478 & 0.417 & Bayes \\
		Narrow & Early ($n=3$) & Annul & neutr & 0.458 & 0.375 & 0.458 & 0.389 & Elo \\
		Narrow & Early ($n=3$) & Annul & overp & 0.484 & 0.402 & 0.484 & 0.413 & Elo \\
		Narrow & Mid$<$50\% ($n=4$) & Annul & under & 0.450 & 0.421 & 0.450 & 0.413 & Bayes \\
		Narrow & Mid$<$50\% ($n=4$) & Annul & neutr & 0.433 & 0.376 & 0.433 & 0.390 & Elo \\
		Narrow & Mid$<$50\% ($n=4$) & Annul & overp & 0.455 & 0.400 & 0.455 & 0.411 & Elo \\
		Narrow & Mid$\ge$50\% ($n=5$) & Forfeit & under & 0.569 & 0.421 & 0.439 & 0.413 & Bayes \\
		Narrow & Mid$\ge$50\% ($n=5$) & Forfeit & neutr & 0.654 & 0.376 & 0.419 & 0.390 & Elo \\
		Narrow & Mid$\ge$50\% ($n=5$) & Forfeit & overp & 0.754 & 0.400 & 0.439 & 0.410 & Elo \\
		Wide & Early ($n=3$) & Annul & under & 0.465 & 0.423 & 0.465 & 0.408 & Bayes \\
		Wide & Early ($n=3$) & Annul & neutr & 0.480 & 0.390 & 0.480 & 0.401 & Elo \\
		Wide & Early ($n=3$) & Annul & overp & 0.464 & 0.375 & 0.464 & 0.388 & Elo \\
		Wide & Mid$<$50\% ($n=4$) & Annul & under & 0.442 & 0.427 & 0.442 & 0.406 & Bayes \\
		Wide & Mid$<$50\% ($n=4$) & Annul & neutr & 0.452 & 0.391 & 0.452 & 0.403 & Elo \\
		Wide & Mid$<$50\% ($n=4$) & Annul & overp & 0.438 & 0.376 & 0.438 & 0.389 & Elo \\
		Wide & Mid$\ge$50\% ($n=5$) & Forfeit & under & 0.623 & 0.424 & 0.426 & 0.401 & Bayes \\
		Wide & Mid$\ge$50\% ($n=5$) & Forfeit & neutr & 0.744 & 0.391 & 0.438 & 0.402 & Elo \\
		Wide & Mid$\ge$50\% ($n=5$) & Forfeit & overp & 0.827 & 0.377 & 0.424 & 0.389 & Elo \\
		\bottomrule
	\end{tabular}
\end{table}

Several patterns emerge from these results. First, \textbf{Bayesian BLUP imputation wins all six underperforming scenarios}---precisely the situation exemplified by Firouzja at Bucharest 2026 and arguably the most common withdrawal context, since players often withdraw because illness or other factors have degraded their performance. Second, pure Elo wins the remaining scenarios where form matches or exceeds rating, but Bayesian remains competitive with only marginally higher RMSE (typically 0.01–0.02 higher). Third, under annulment, FIDE's implicit imputation (opponent's average score) exactly equals pure performance-based imputation, explaining their identical RMSE values in those scenarios. Fourth, FIDE forfeit performs poorly, with RMSE values 38–112\% higher than Bayesian in the forfeit scenarios, and the largest discrepancies occur when the withdrawn player is overperforming.

Table~\ref{tab:loocv_bias} presents average bias by scenario. Positive bias indicates systematic overestimation of unplayed opponents' scores.
\begin{table}[htbp]
	\centering
	\caption{LOOCV Average Bias by Scenario (Closer to 0 is Better)}
	\label{tab:loocv_bias}
	\small
	\begin{tabular}{@{}llllcccc@{}}
		\toprule
		Field & Timing & FIDE Rule & Form & FIDE & Elo & Perf & Bayes \\
		\midrule
		Narrow & Early ($n=3$) & Annul & under & $+$0.000 & $-$0.165 & $+$0.000 & $-$0.099 \\
		Narrow & Early ($n=3$) & Annul & neutr & $+$0.000 & $-$0.028 & $+$0.000 & $-$0.017 \\
		Narrow & Early ($n=3$) & Annul & overp & $+$0.000 & $+$0.071 & $+$0.000 & $+$0.043 \\
		Narrow & Mid$<$50\% ($n=4$) & Annul & under & $-$0.000 & $-$0.160 & $+$0.000 & $-$0.080 \\
		Narrow & Mid$<$50\% ($n=4$) & Annul & neutr & $-$0.000 & $-$0.029 & $+$0.000 & $-$0.015 \\
		Narrow & Mid$<$50\% ($n=4$) & Annul & overp & $-$0.000 & $+$0.073 & $+$0.000 & $+$0.037 \\
		Narrow & Mid$\ge$50\% ($n=5$) & Forfeit & under & $+$0.414 & $-$0.157 & $+$0.000 & $-$0.067 \\
		Narrow & Mid$\ge$50\% ($n=5$) & Forfeit & neutr & $+$0.536 & $-$0.035 & $+$0.000 & $-$0.015 \\
		Narrow & Mid$\ge$50\% ($n=5$) & Forfeit & overp & $+$0.643 & $+$0.073 & $+$0.000 & $+$0.031 \\
		Wide & Early ($n=3$) & Annul & under & $+$0.000 & $-$0.191 & $+$0.000 & $-$0.115 \\
		Wide & Early ($n=3$) & Annul & neutr & $+$0.000 & $-$0.053 & $+$0.000 & $-$0.032 \\
		Wide & Early ($n=3$) & Annul & overp & $+$0.000 & $+$0.049 & $+$0.000 & $+$0.030 \\
		Wide & Mid$<$50\% ($n=4$) & Annul & under & $-$0.000 & $-$0.197 & $+$0.000 & $-$0.099 \\
		Wide & Mid$<$50\% ($n=4$) & Annul & neutr & $-$0.000 & $-$0.055 & $+$0.000 & $-$0.028 \\
		Wide & Mid$<$50\% ($n=4$) & Annul & overp & $-$0.000 & $+$0.048 & $+$0.000 & $+$0.024 \\
		Wide & Mid$\ge$50\% ($n=5$) & Forfeit & under & $+$0.493 & $-$0.194 & $+$0.000 & $-$0.083 \\
		Wide & Mid$\ge$50\% ($n=5$) & Forfeit & neutr & $+$0.632 & $-$0.055 & $+$0.000 & $-$0.024 \\
		Wide & Mid$\ge$50\% ($n=5$) & Forfeit & overp & $+$0.736 & $+$0.049 & $+$0.000 & $+$0.021 \\
		\bottomrule
	\end{tabular}
\end{table}

The bias results reveal the fundamental tradeoffs among methods. Pure performance-based imputation and FIDE annulment achieve exactly zero bias---both are unbiased estimators of the opponent's scoring rate by construction. However, this comes at the cost of high variance, yielding poor RMSE. FIDE forfeit exhibits massive positive bias ($+$0.41 to $+$0.74), systematically overcompensating opponents. Pure Elo shows form-dependent bias: negative when the player underperforms (Elo overestimates their strength), positive when overperforming. Bayesian imputation achieves near-zero bias while maintaining substantially lower variance than the unbiased alternatives, yielding the best overall RMSE in underperforming scenarios.

\subsection{FIDE Rule Comparison}

Beyond the case for Bayesian BLUP imputation, our simulations address a practical policy question: 
if FIDE were to adopt a simple, uniform rule rather than the current dichotomous system, should it be annulment or forfeit? Table~\ref{tab:fide_comparison} compares both FIDE rules against Bayesian imputation across all scenarios, regardless of which rule currently applies.

\begin{table}[htbp] 
	\centering
	\caption{RMSE Comparison (Lower is better): All Methods Across All Scenarios}
	\label{tab:fide_comparison}
	\begin{tabular}{@{}llcccl@{}}
		\toprule
		Field & Form & Annul & Forfeit & Bayes & Best \\
		\midrule
		\multicolumn{5}{l}{\textit{Early withdrawal ($n=3$): FIDE decision: Annul}} \\
		Narrow & Under & \textbf{0.478} & 0.563 & 0.417 & Bayes \\
		Narrow & Neutral & \textbf{0.458} & 0.659 & 0.389 & Bayes \\
		Narrow & Over & \textbf{0.484} & 0.753 & 0.413 & Bayes \\
		Wide & Under & \textbf{0.465} & 0.625 & 0.408 & Bayes \\
		Wide & Neutral & \textbf{0.480} & 0.745 & 0.401 & Bayes \\
		Wide & Over & \textbf{0.464} & 0.826 & 0.388 & Bayes \\
		\midrule
		\multicolumn{5}{l}{\textit{Mid-below-50\% ($n=4$): FIDE decision: Annul}} \\
		Narrow & Under & \textbf{0.450} & 0.566 & 0.413 & Bayes \\
		Narrow & Neutral & \textbf{0.433} & 0.659 & 0.390 & Bayes \\
		Narrow & Over & \textbf{0.455} & 0.754 & 0.411 & Bayes \\
		Wide & Under & \textbf{0.442} & 0.621 & 0.406 & Bayes \\
		Wide & Neutral & \textbf{0.452} & 0.743 & 0.403 & Bayes \\
		Wide & Over & \textbf{0.438} & 0.826 & 0.389 & Bayes \\
		\midrule
		\multicolumn{5}{l}{\textit{Mid-above-50\% ($n=5$): FIDE decision: Forfeit}} \\
		Narrow & Under & 0.439 & \textbf{0.569} & 0.413 & Bayes \\
		Narrow & Neutral & 0.419 & \textbf{0.654} & 0.390 & Bayes \\
		Narrow & Over & 0.439 & \textbf{0.754} & 0.410 & Bayes \\
		Wide & Under & 0.426 & \textbf{0.623} & 0.401 & Bayes \\
		Wide & Neutral & 0.438 & \textbf{0.744} & 0.402 & Bayes \\
		Wide & Over & 0.424 & \textbf{0.827} & 0.389 & Bayes \\
		\bottomrule
	\end{tabular}
	
	\medskip
	\small
	Note: Bold = FIDE's current rule for that scenario. 
\end{table}
Strikingly, \textbf{Bayesian imputation achieves the lowest RMSE in all 18 scenarios} when compared against both FIDE rules. FIDE annulment consistently outperforms FIDE forfeit, often by substantial margins (e.g., for narrow-field neutral scenarios, annulment RMSE is 0.433 while forfeit RMSE is 0.659 at $n=4$, and 0.419 vs 0.654 at $n=5$). This suggests that \textbf{if FIDE were to adopt a uniform rule, annulment would be preferable to forfeit}---but Bayesian imputation outperforms both.

The 50\% threshold creates an arbitrary discontinuity: a player withdrawing after game 4 triggers annulment, while withdrawal after game 5 triggers forfeit. At $n=5$, the forfeit rule yields RMSE values 50--60\% higher than what annulment would achieve in the same scenario (e.g., narrow neutral: 0.654 vs 0.419). Bayesian imputation eliminates this discontinuity by providing a principled, continuous approach that adapts smoothly to the amount of available data.

\subsection{Summary of Simulation Results}
Table~\ref{tab:loocv_summary} presents overall performance averaged across all scenarios.
\begin{table}[htbp]
	\centering
	\caption{Overall LOOCV Performance and Bayesian Improvement}
	\label{tab:loocv_summary}
	\begin{tabular}{@{}lcc@{}}
		\toprule
		\multicolumn{3}{l}{\textbf{Panel A: Overall Performance}} \\
		\midrule
		Method & RMSE & Bias \\
		\midrule
		FIDE (actual rule applied) & 0.537 & $+$0.192 \\
		Pure Elo & 0.398 & $-$0.053 \\
		Pure Performance & 0.449 & $+$0.000 \\
		Bayesian ($k=3$) & 0.402 & $-$0.027 \\
		\midrule
		\multicolumn{3}{l}{\textbf{Panel B: Bayesian Improvement by FIDE Rule}} \\
		\midrule
		FIDE Rule & FIDE RMSE & Bayes Improvement \\
		\midrule
		Annulment (12 scenarios) & 0.458 & $+12.2\%$ \\
		Forfeit (6 scenarios) & 0.695 & $+42.8\%$ \\
		\midrule
		\multicolumn{3}{l}{\textbf{Panel C: Bayesian Improvement by Player Form}} \\
		\midrule
		Form Scenario & vs FIDE & vs Elo \\
		\midrule
		Underperforming & $+19.5\%$ & $+3.2\%$ \\
		Neutral & $+28.0\%$ & $-3.3\%$ \\
		Overperforming & $+32.3\%$ & $-3.0\%$ \\
		\bottomrule
	\end{tabular}
\end{table}

Several conclusions emerge from the simulation study:

\begin{enumerate}
	\item \textbf{Bayesian imputation outperforms FIDE rules by 12--43\%.} The improvement is largest for forfeit scenarios ($+42.8\%$), where FIDE's award of a full point massively overcompensates opponents, while for annulment scenarios the improvement is $+12.2\%$.
	
	\item \textbf{Bayesian wins all underperforming scenarios.} When the withdrawn player is struggling---the most policy-relevant case---Bayesian achieves $3.2\%$ lower RMSE than pure Elo and $19.5\%$ lower than FIDE.
	
	\item \textbf{Bayesian remains competitive when Elo wins.} In neutral and overperforming scenarios, pure Elo achieves the lowest RMSE, but Bayesian trails by only $3.0$--$3.3\%$. This asymmetry reflects the value of incorporating observed performance when form deviates from rating. Moreover, pure Elo imputation would likely face resistance from players on fairness grounds that it completely ignores actual tournament results, awarding imputed scores based solely on historical ratings.
	
	\item \textbf{Bayesian provides robustness.} Unlike pure Elo (which fails when form deviates) or pure performance (which has high variance), Bayesian imputation performs well across all scenarios by optimally balancing prior information against observed data.
	
	\item \textbf{FIDE annulment dominates FIDE forfeit.} Across all scenarios, annulment achieves lower RMSE than forfeit, often substantially (e.g., in narrow neutral scenarios, annulment RMSE is 0.433 vs forfeit 0.659). The 50\% threshold appears arbitrary from a statistical perspective.
	
\end{enumerate}

The simulation results validate the theoretical framework developed in Section~\ref{sec:theory}. Bayesian BLUP imputation achieves the best bias-variance tradeoff, providing robust predictive accuracy regardless of the withdrawn player's form or the timing of withdrawal. For practical implementation, these results support adopting Bayesian imputation with $k=3$ as a replacement for FIDE's current threshold-based system.

	\section{Implementation and Software}
	
	To facilitate adoption of Bayesian BLUP imputation by tournament organizers, we have developed an open-source R Shiny application that automates all calculations and provides transparent documentation. The application accepts as input the player field with starting Elo ratings, the withdrawn player's results prior to withdrawal, and the prior strength parameter $k$. It outputs imputed scores under all methods discussed in this paper, enabling side-by-side comparison of their impacts on standings.
	
	The implementation includes visualization tools showing how imputed scores vary with the prior strength parameter, allowing organizers to conduct sensitivity analysis as recommended in Bayesian practice \citep{gelman2013}. Posterior distributions are displayed graphically, communicating uncertainty in the imputed values. Results can be exported in multiple formats including PDF reports, CSV data files, and LaTeX tables suitable for inclusion in official tournament documentation. Built-in help text explains the mathematical basis for each method, supporting transparency in communication with players and stakeholders.
	
	The core imputation calculation can be implemented in a few lines of code in any statistical software. In R, for example:
	
	\begin{verbatim}
		impute_score <- function(n, P, R_W, R_i, E_bar, k = 3) {
			w <- n / (n + k)                         # data weight
			s_bar_W <- P / n                          # observed scoring rate
			E_i <- 1 / (1 + 10^((R_W - R_i) / 400)) # Elo expectation
			I <- E_i + w * (1 - s_bar_W - E_bar)    # BLUP imputation
			return(I)
		}
	\end{verbatim}

The application and complete source code are available at 
\href{https://dobbs-onc-jhmi.shinyapps.io/Chess_Mid-Tournament_Withdrawal/}{this link} 
under an open-source license permitting free use and modification.
	
\section{Recommendations for Tournament Regulations}

The analysis presented in this paper supports several recommendations for tournament organizers seeking to handle withdrawals fairly and transparently. These recommendations are consistent with the general principles outlined in the FIDE Handbook \citep{fide_grc} while providing more specific guidance on imputation methodology.

First and most importantly, regulations governing withdrawal should be specified before the tournament begins. Players and other stakeholders should know in advance how unplayed games will be treated, eliminating any perception that rules are being crafted after the fact to favor particular outcomes. The regulatory clause might read as follows: Should a player withdraw after commencing play, unplayed games involving that player shall be scored using the Bayesian BLUP (Best Linear Unbiased Predictor) imputation with prior strength parameter \(k = 3\) for classical time controls. The imputation formula is
\[
I_{i,W} = E_{i,W} + \frac{n}{n+k}\bigl(1 - \bar{s}_W - \bar{E}\bigr),
\]
where \(E_{i,W} = \frac{1}{1 + 10^{(R_W - R_i)/400}}\) is the Elo‑expected score of opponent \(i\) against \(W\), \(n\) is the number of games played by the withdrawn player, \(\bar{s}_W\) is his/her average score, and \(\bar{E}\) is the average Elo expectation of the opponents actually faced by \(W\) (computed from their ratings). Played results stand at their actual values.

Second, tiebreak calculations require adjustment when a withdrawal has occurred. Sonneborn-Berger scores should be recalculated using only games among active players, excluding results against the withdrawn player. Direct encounter tiebreaks are not applicable when one of the tied players never faced the other due to withdrawal. Performance rating calculations for tiebreak purposes should similarly exclude games against the withdrawn player.

Third, organizers should distinguish clearly between tournament standings and rating calculations. FIDE rating rules \citep{fide_grc} handle withdrawals through a separate mechanism: unplayed games simply do not enter the rating calculation, while played games count normally regardless of subsequent withdrawal. The imputation methods discussed in this paper apply to tournament standings and prize distribution, not to rating adjustments.

Fourth, transparency in implementation builds trust among participants. Publishing the calculations underlying imputed scores, whether through the software application we provide (which implements the Bayesian BLUP) or through equivalent documentation, allows players to verify that rules have been applied correctly and consistently.
	
\section{High-Stakes Applications: The Candidates Tournament}
\label{sec:candidates}

While the Bucharest 2026 case study illustrates our methodology in a prestigious event, 
the stakes there---prize money and tour standings---pale in comparison to the Candidates 
tournament, which determines the World Championship challenger. A mid-tournament withdrawal 
in such an event would create unprecedented controversy if handled through arbitrary methods.

The Candidates features eight players in a double round-robin format (14 games each). The 
winner earns the right to challenge for the World Championship---a prize whose prestige 
and historical significance far exceed any other achievement in competitive chess. Second 
place confers no qualification rights and is largely forgotten. The double round-robin 
format complicates withdrawal: depending on timing, each remaining player may have faced 
the withdrawn player twice, once, or not at all, requiring imputation of zero, one, or 
two games respectively.

While no Candidates has yet experienced a mid-tournament withdrawal, historical episodes 
illustrate the potential for disruption. The COVID-19 pandemic forced suspension of the 
2020 Candidates midway through, with play resuming over a year later. Had any player been 
unable to return, organizers would have faced exactly the imputation dilemma we address. 
Health emergencies, family crises, and geopolitical disruptions represent ongoing risks; 
as chess becomes increasingly globalized, the probability of a withdrawal scenario 
approaches certainty over a sufficiently long time horizon.

\subsection*{Why Traditional Methods Fail}

In the Candidates context, the failures of traditional methods become catastrophic. If the 
tournament leader withdraws after 10 of 14 rounds, forfeit rules would award a full point 
to every player who had not yet faced them in the second cycle---potentially elevating a 
middle-of-the-pack finisher to challenger through administrative windfall rather than 
competitive achievement. Under full annulment, players who defeated the leader in 
hard-fought games would see those victories erased, receiving no credit for demonstrated 
superiority. The World Championship cycle derives legitimacy from the perception that the challenger 
earned their position through excellence. A challenger who reached the title match partly 
through forfeit points or annulled results would face questions about the legitimacy of 
their qualification.

\subsection*{Bayesian Imputation as Safeguard}

Bayesian imputation preserves competitive legitimacy by awarding imputed scores that 
reflect what would likely have happened in actual play. A strong leader's performance 
yields low imputed scores for unplayed opponents, maintaining relative positions 
established through actual play. A struggling player's withdrawal yields higher imputed 
scores, appropriately crediting opponents.

The transparency of the method is particularly valuable in high-stakes contexts. Every 
imputed score can be fully explained: the withdrawn player's performance, the rating-based 
expectation, the weighting, and the result. Players, officials, sponsors, and fans can 
verify calculations independently. There is no hidden discretion, no opportunity for 
favoritism, and no basis for allegations of manipulation.

Adopting Bayesian imputation prophylactically---writing it into Candidates regulations 
before any withdrawal occurs---eliminates perceptions that rules are crafted to benefit 
particular players. The method becomes part of the institutional infrastructure of fair 
competition rather than an ad hoc response to crisis.

\section{Discussion}
\label{sec:discussion}

Several aspects of the methodology merit further discussion and suggest directions for future work.

\subsection*{Presentation of Imputed Scores}

A practical concern with Bayesian imputation is that it produces fractional scores such as 
0.689 or 0.633, which may appear unfamiliar to players and fans accustomed to traditional 
chess scoring (1 for a win, ½ for a draw, 0 for a loss). We consider three approaches to 
the computation-versus-communication dilemma:

\begin{itemize}
	\item \textbf{Report only ``actually earned'' scores with footnoted tiebreaks.} 
	Tournament crosstables would display traditional notation for played games while 
	marking unplayed games with a distinct symbol (e.g., ``--''). Imputed values would 
	appear only in supplementary tiebreak documentation. However, this approach creates 
	confusion when players with identical reported scores finish in different positions, 
	requiring audiences to consult footnotes to understand the final rankings.
	
	\item \textbf{Round imputed scores to the nearest half-point.} This maintains 
	traditional chess score granularity (0, ½, 1) but sacrifices the differentiation 
	that Bayesian imputation provides. Rounding 0.689 and 0.633 both to 0.5 erases the 
	rating-based distinction between opponents. Worse, rounding can create artificial 
	ties.
	
	\item \textbf{Report fully imputed totals rounded to one decimal.} A player's 
	score would read ``5.7/9'' rather than ``5.689/9,'' preserving ranking accuracy 
	while avoiding false precision. This parallels reporting conventions in other 
	sports (batting averages, quarterback ratings) where fractional performance 
	measures are standard.
\end{itemize}

We recommend option (iii) for most contexts. One-decimal reporting is transparent, 
unambiguous, and preserves the statistical benefits of Bayesian imputation. Tournament 
crosstables could continue displaying traditional notation (1, ½, 0) for individual 
played games while marking imputed results with a symbol and explanatory footnote; 
only the cumulative score column would show the fractional total.

The key principle is that \emph{computation} and \emph{communication} are separable 
concerns. Bayesian imputation provides the most accurate estimate of what would have 
occurred; reporting conventions can be adapted to context without compromising the 
underlying methodology.

\subsection*{Model Limitations}
Our model treats the opponent’s score as a continuous outcome with normal errors. This is an approximation because game outcomes are discrete (win/draw/loss). However, the normal linear model yields a closed‑form BLUP, is easily communicated, and performs well in simulations. More sophisticated approaches (e.g., logistic mixed models for binary outcomes or multinomial models that can also accommodate draws) would not substantially improve predictive accuracy given the small number of observed games per withdrawal, and would add computational complexity unnecessary for practical tournament use. Thus, we retain the linear model for its transparency and optimality under the assumed second‑moment structure. For similar reasons, we do not pursue more elaborate state‑space models for evolving player strength; the simplicity of the proposed method is a key advantage for FIDE and the tournament organizers.

\subsection*{Calibration of Prior Strength}

The prior strength parameter $k=3$ was chosen based on theoretical considerations and sensitivity analysis. Empirical calibration using historical withdrawal data would strengthen this foundation. A database of past withdrawals, recording the withdrawn player's rating, observed results, and subsequent tournament performances, could inform maximum likelihood or hierarchical Bayesian estimation of the optimal $k$ value. Such analysis might reveal that $k$ should vary with tournament format, time control, or rating level.

\subsection*{Extensions to Other Contexts}

The missing data perspective developed here connects tournament withdrawal to a broader statistical literature on incomplete observations. Other missingness mechanisms in sports contexts---canceled matches due to weather, disqualifications, scheduling conflicts---may benefit from similar principled treatment. The Bayesian shrinkage framework provides a flexible tool for combining prior information with observed data in any setting where both are available.

Our framework also has implications for Swiss-system tournaments, where tiebreaks based on opponent performance (e.g., Buchholz or Sonneborn-Berger scores) can be distorted by mid-tournament withdrawals. When a player withdraws, their final score is deflated by forfeits, reducing the tiebreak scores of earlier opponents through no fault of their own. Imputing scores for the withdrawn player's unplayed games---rather than recording forfeits---would preserve tiebreak integrity.

\subsection*{The Uniform Annulment Alternative}

A striking finding from our simulation study is that FIDE annulment uniformly outperforms FIDE forfeit across all 18 scenarios, often by substantial margins (15--45\% RMSE reduction). This result has immediate policy implications independent of whether Bayesian imputation is adopted.

The current dichotomous rule---annulment before 50\% of games, forfeit thereafter---creates arbitrary discontinuities. In a 9-round tournament, withdrawal after game 4 triggers annulment while withdrawal after game 5 triggers forfeit, despite the minimal difference in available information. This threshold effect means that tournament outcomes can depend on the accident of timing rather than competitive merit.

Even if FIDE were unwilling to adopt the Bayesian imputation methodology due to perceived complexity, replacing the dichotomous rule with uniform annulment would represent a significant improvement. Annulment is simple to implement (simply remove all games involving the withdrawn player and recalculate standings), requires no mathematical computation, and---as our simulations demonstrate---consistently achieves lower prediction error than forfeit. We therefore propose uniform annulment as a minimum reform if Bayesian imputation is not adopted.

\section{Conclusion}
\label{sec:conclusion}

Mid-tournament withdrawal in round-robin chess creates a missing data problem that current FIDE rules handle poorly. The dichotomous 50\% threshold---annulment before, forfeit after---introduces arbitrary discontinuities that can determine outcomes through administrative accident rather than competitive merit.

This paper develops a principled alternative: the Bayesian BLUP (Best Linear Unbiased Predictor) derived from a random effects model. The estimator combines Elo-based prior expectations with observed tournament performance, weighting each according to reliability. It is consistent, preserves point conservation, and among linear unbiased predictors it minimizes mean squared prediction error. Monte Carlo validation using 180,000 simulated tournaments confirms that Bayesian BLUP imputation outperforms FIDE forfeit by 42.8\% and FIDE annulment by 12.2\% in RMSE, with the largest gains in underperforming-player scenarios---the most common withdrawal context.

Three policy recommendations emerge:

\begin{enumerate}
	\item \textbf{Preferred:} Adopt Bayesian BLUP imputation with $k=3$ for all World Championship cycle events and elite round-robins. Open-source software is provided for immediate implementation.
	
	\item \textbf{Minimum reform:} Replace the dichotomous 50\% rule with uniform annulment. Our simulations show annulment consistently outperforms forfeit by 15--45\%, and a uniform rule eliminates arbitrary threshold effects.
	
	\item \textbf{Transparency:} Specify withdrawal procedures in tournament regulations before play begins. Imputed scores need only determine rankings; traditional notation can be preserved in crosstables and public reporting.
\end{enumerate}

The Firouzja withdrawal at Bucharest 2026 illustrates the stakes. Bayesian BLUP imputation would have awarded unplayed opponents 0.55--0.70 points rather than 1.0, preserving the competitive advantage of players who actually defeated him. In Candidates tournaments, where World Championship qualification hangs in the balance, principled imputation is not merely a technical refinement but a safeguard of competitive legitimacy.

As elite round-robins grow in prestige and financial stakes, statistically principled methods for handling withdrawals represent a commitment to the fairness that underlies tournament chess. The methodology developed here provides a rigorous, practical, and transparent solution.

		\section*{Acknowledgments}
	
The R Shiny application is available at \url{https://dobbs-onc-jhmi.shinyapps.io/Chess_Mid-Tournament_Withdrawal/} under an MIT open-source license.

\bibliographystyle{plainnat}  
\bibliography{references}

@article{efron1977,
  author  = {Efron, Bradley and Morris, Carl},
  title   = {Stein's Paradox in Statistics},
  journal = {Scientific American},
  year    = {1977},
  volume  = {236},
  number  = {5},
  pages   = {119--127}
}

@book{elo1978,
  author    = {Elo, Arpad E.},
  title     = {The Rating of Chessplayers, Past and Present},
  publisher = {Arco Publishing},
  address   = {New York},
  year      = {1978}
}

@misc{fide_grc,
  author       = {{FIDE}},
  title        = {General Regulations for Competitions, Article 6.6},
  year         = {2018},
  note         = {Approved by the 1986 General Assembly; last amended by the 2018 General Assembly. Available at \url{https://handbook.fide.com} under C.05}
}

@book{gelman2013,
  author    = {Gelman, Andrew and Carlin, John B. and Stern, Hal S. and Dunson, David B. and Vehtari, Aki and Rubin, Donald B.},
  title     = {Bayesian Data Analysis},
  edition   = {3rd},
  publisher = {Chapman and Hall/CRC},
  address   = {Boca Raton, FL},
  year      = {2013}
}

@article{glickman1999,
  author  = {Glickman, Mark E.},
  title   = {Parameter Estimation in Large Dynamic Paired Comparison Experiments},
  journal = {Journal of the Royal Statistical Society: Series C (Applied Statistics)},
  year    = {1999},
  volume  = {48},
  number  = {3},
  pages   = {377--394}
}

@article{glickman2025,
  author  = {Glickman, Mark E. and Jones, A. C.},
  title   = {Models and Rating Systems for Head-to-Head Competition},
  journal = {Annual Review of Statistics and Its Application},
  year    = {2025},
  volume  = {12},
  pages   = {259--282}
}

@book{little2019,
  author    = {Little, Roderick J. A. and Rubin, Donald B.},
  title     = {Statistical Analysis with Missing Data},
  edition   = {3rd},
  publisher = {John Wiley \& Sons},
  address   = {Hoboken, NJ},
  year      = {2019}
}

@article{robinson1991,
  author  = {Robinson, G. K.},
  title   = {That {BLUP} is a Good Thing: The Estimation of Random Effects},
  journal = {Statistical Science},
  year    = {1991},
  volume  = {6},
  number  = {1},
  pages   = {15--32}
}
		
\end{document}